\documentclass[%
 reprint,
 amsmath,amssymb,
 aps,
]{revtex4-1}
\usepackage{amsthm}
\usepackage{graphicx}
\usepackage{dcolumn}
\usepackage{bm}
\usepackage{hyperref}

\usepackage[caption=false]{subfig}

\newtheorem{theorem}{Theorem}

\DeclareMathOperator{\diag}{diag}
\newcommand{\GG}{\mathcal{G}}

\newcommand{\E}{\mathcal{E}}

\newcommand{\ones}{\mathbf{1}}

\begin{document}

\preprint{APS/123-QED}

\title[Maximizing Algebraic Connectivity in Multiplex Networks]{Maximizing Algebraic Connectivity in Interconnected Networks}
\thanks{This material is based upon work supported by the National Science Foundation under Grants No. (DMS-1515810 and partially by CIF-1423411)}
\email{Corresponding author: heman@ksu.edu} 
\author{Heman Shakeri$^{1}$}
\author{Nathan Albin$^2$}%
\author{Faryad Darabi Sahneh$^1$}
\author{Pietro Poggi-Corradini$^2$}
\author{Caterina Scoglio$^1$}
\affiliation{$^1$Electrical and Computer Engineering Department, Kansas State University, Manhattan, Kansas, USA}\affiliation{$^2$Mathematics Department, Kansas State University, Manhattan, Kansas, USA}
\date{\today}

\begin{abstract}
Algebraic connectivity, the second eigenvalue of the Laplacian matrix, is a measure of node and link connectivity on networks. When studying interconnected networks it is useful to consider a multiplex model, where the component networks operate together with inter-layer links among them. In order to have a well-connected multilayer structure, it is necessary to optimally design these inter-layer links considering realistic constraints. In this work, we solve the problem of finding an optimal weight distribution for one-to-one inter-layer links
under budget constraint. 
We show that for the special multiplex configurations with identical layers, the uniform weight distribution is always optimal. 
On the other hand, when the two layers are arbitrary, increasing the budget reveals the existence of two different regimes. Up to a certain threshold budget, the second eigenvalue of the supra-Laplacian is simple, the optimal weight distribution is uniform, and the Fiedler vector is constant on each layer.  Increasing the budget past the threshold, the optimal weight distribution can be non-uniform. The interesting consequence of this result is that there is no need to solve the optimization problem when the available budget is less than the threshold, which can be easily found analytically.
\end{abstract}

\maketitle




Real-world networks are often connected together and therefore influence each other \cite{Arenas2014multilayer}. Robust design of interdependent networks is critical to allow uninterrupted flow of information, power, and goods in spite of possible errors and attacks \cite{Buldyrev2010}.
The second eigenvalue of the Laplacian matrix, $\lambda_2(L)$, is a good measure of network robustness \cite{Jamakovic2007}.
Fiedler shows that algebraic connectivity increases by adding links \cite{Fiedler1973}.
Moreover, it is harder to bisect a network with higher algebraic connectivity \cite{Fallat2003}. 
%
The second eigenvalue of the Laplacian matrix is also a measure of the speed of mixing for a Markov process on a network \cite{bremaud2013markov}.
Boyd et al. maximize the mixing rate by assigning optimum link weights in the setting of a single layer (\cite{boyd2004fastestChain} and \cite{Boyd2006FastestMP}).
%

For multiplex networks (see Fig. \ref{fig:multilayer}), a natural question is the following. Given fixed network layers, how should the weights be assigned to inter-layer links in order to maximize algebraic connectivity? 
\begin{figure}[b]
\includegraphics[scale=1]{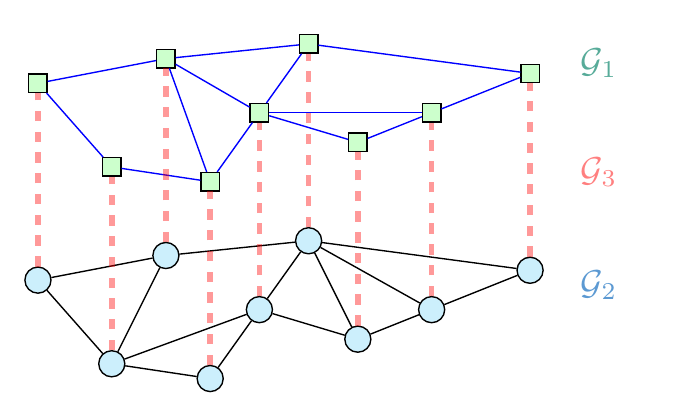}
\caption{A schematic of a multiplex network $\GG$ with two layers $\GG_1$, $\GG_2$, connecting through an inter-layer one-to-one structure $\GG_3$.}
\label{fig:multilayer}
\end{figure}

The behavior of $\lambda_2$, in the case of identical weights, i.e., with a fixed coupling weight $p$ for every inter-layer link, has been studied recently.
For instance, Gomez et al. observe that $\lambda_2(L)$ grows linearly with $p$ up to a critical $p^*$, and then has a non-linear behavior afterwards \cite{gomez2013diffusion}.
%
Radicchi and Arenas find bounds for this threshold value $p^*$ \cite{Radicchi2013}. Sahneh et al. compute the exact value analytically \cite{sahneh2014exact}.

Martin-Hernandez et al. analyze the algebraic connectivity and Fiedler vector of multiplex structures, with addition of a number of inter-layer links in two configurations;  diagonal (one-to-one) and random \cite{Gregorio2014algebraic}. They show that for the first case, algebraic connectivity saturates after adding a sufficient number of links.
Li et al. adopt a network flow approach to propose a heuristic that improves robustness of large multiplex networks by choosing from a set of inter-layer links with predefined weights \cite{Xin2015connectivity}. 

In this letter we remove the constraint of identical interlinking weights and pose the problem of finding the maximum algebraic connectivity for  a one-to-one interconnected structure between different layers in the presence of limited resources. 
We show that up to the threshold budget $p^*N$---where $p^*$ is the same threshold studied in \cite{gomez2013diffusion}, \cite{Radicchi2013}, and \cite{sahneh2014exact}---the uniform distribution of identical weights is actually optimal. For larger budgets, the optimal distribution of weights is generally not uniform.

Let $\mathcal{G}=\left( \mathcal{V},\mathcal{E}\right) $
represents a network and by $\mathcal{V}=\left\{ 1,\ldots,N\right\} $
and $\mathcal{E}\subset\mathcal{V}\times\mathcal{V}$, we denote the set of nodes and links. 
For a link $e$ between nodes $u$ and $v$, i.e, $e:\{u,v\}\in\mathcal{E}$, we define a nonnegative value $w_{uv}$ as the weight of the link.
The Laplacian matrix of $\GG$ can be defined as:
\begin{equation}\label{eq:LapRepre}
L= \sum_{ij\in \E}w_{ij}B_{ij}
\end{equation}
where $B_{ij}:= (e_i-e_j)(e_i-e_j)^T$ is the incidence matrix for link $ij$, and $e_i$ is a vector with $i$th component one and rest of its elements are zero.

For a multiplex network with two layers $\GG_1 = \left\{ \mathcal{V}_1,\mathcal{E}_1\right\}$ and $\GG_2=\left\{ \mathcal{V}_2,\mathcal{E}_2\right\}$ and $|\mathcal{V}_1| =| \mathcal{V}_2|$, we consider a bipartite graph $\mathcal{G}_3=\left\{ \mathcal{V},\mathcal{E}_3\right\}$ with $\mathcal{E}_3\subseteq \mathcal{V}_1\times \mathcal{V}_2$. The multiplex network $\GG$ is composed from $\GG_1$, $\GG_2$, and $\GG_3$ (Fig. \ref{fig:multilayer}).
We want to design optimal weights for $\GG_3$ to improve the algebraic connectivity of $\GG$ as much as possible with a limited budget, i.e., $\sum w_{ij} = c$.
Using Eq. \eqref{eq:LapRepre}, the Laplacian matrix of $\GG$ (supra-Laplacian matrix), is:
\begin{equation}
\label{eq:Laplacian}
L(w)=
\sum_{ij\in \E_2\cup \E_3}B_{ij}+\sum_{ij\in \E_3}w_{ij}B_{ij},
\end{equation}
where we use the notation $L(w)$ to make explicit the dependence of the Laplacian on the interlayer weights $w$.
From Eq. \eqref{eq:Laplacian}, the Laplacian, $L$, of the combined network takes the form

\begin{equation*}
  L(w) =
  \begin{bmatrix}
    L_1 & 0\\
    0 & L_2
  \end{bmatrix}
  + 
  \begin{bmatrix}
    W & -W\\
    -W & W
  \end{bmatrix},
\end{equation*}
where $L_1$ and $L_2$ are the Laplacians of the individual layers and $W=\diag(w)$ with $w\ge 0$ the inter-layer link weights satisfying the budget constraint $w^T\ones = c$.  We assume the two layers are connected independently, so that $\lambda_3(L)>0$, for all choices of $c$ and $w$.

The second eigenvalue can be characterized as the solution to to the optimization problem
\begin{equation}
  \label{eq:lambda2}
  \lambda_2(L) = \min_{\substack{v\ne 0\\v^T\ones=0}} \frac{v^TLv}{\|v\|^2}.
\end{equation}
The optimal weight problem, then, can be phrased as follows.  Given a budget $c\ge 0$, solve the problem
\begin{equation}
  \label{eq:F}
  F(c) := \max_{\substack{w\ge 0\\w^T\ones = c}} \lambda_2(L(w)).
\end{equation}

Since $L$ is an affine function of $w$, and $\lambda_2$ is a concave function of $L$, it follows that~\eqref{eq:F} is a convex optimization problem.  In fact, it can be recast as an SDP (similarly to the construction in~\cite{Boyd2006FastestMP}) and, thus, can be solved efficiently even for large networks using standard numerical methods.

Returning to~\eqref{eq:lambda2}, it is convenient to write $v$ in component form $v=(v_1^T,v_2^T)^T$ so that~\eqref{eq:lambda2} implies
\begin{equation}
  \label{eq:main-inequality}
  \begin{split}
  v_1^TL_1v_1 + & v_2^TL_2v_2 + (v_1-v_2)^TW(v_1-v_2) \\
  & - \lambda_2(L)\left(\|v_1\|^2+\|v_2\|^2\right)
  \ge 0 \qquad\forall \  v_1^T\ones = -v_2^T\ones.
  \end{split}
\end{equation}
Since $v$ must satisfy $v_1^T\ones = -v_2^T\ones$, we write $v_1$ and $v_2$ of the form
\begin{equation*}
  v_1 = \alpha\ones + u_1,\qquad v_2 = -\alpha\ones + u_2,\qquad
  u_1^T\ones = u_2^T\ones = 0,
\end{equation*}
for some constant $\alpha$.
Rewriting the terms in~\eqref{eq:main-inequality}, we observe that
\begin{equation*}
  \begin{split}
    (v_1-v_2)^T&W(v_1-v_2) \\
    &=(2\alpha\ones + u_1-u_2)^TW(2\alpha\ones
    + u_1-u_2)\\
    &= 4\alpha^2c + 4\alpha w^T(u_1-u_2) \\
    &~~+ (u_1-u_2)^TW(u_1-u_2)
  \end{split}
\end{equation*}
and that
\begin{equation*}
  \|v_i\|^2 = \|\alpha\ones\|^2 + \|u_i\|^2 = \alpha^2N + \|u_i\|^2\qquad
  \text{for }i=1,2.
\end{equation*}
Thus, Eq.~\eqref{eq:main-inequality} implies that
\begin{equation}
  \label{eq:reduced-inequality}
  \begin{split}
    u_1^TL_1u_1 &+ u_2^TL_2u_2 + 4\alpha^2 c + 4\alpha w^T(u_1-u_2)\\
     &+    (u_1-u_2)^TW(u_1-u_2) - \\
     & \lambda_2(L)\left(2\alpha^2N + \|u_1\|^2 + \|u_2\|^2\right) \ge 0\\
     &    \qquad\forall \alpha, u_1^T\ones=u_2^T\ones=0.
    \end{split}
\end{equation}
Setting $u_1=u_2=0$ in~\eqref{eq:reduced-inequality}, then, gives the inequality
\begin{equation*}
  4\alpha^2c - 2\alpha^2N\lambda_2(L) \ge 0\qquad\forall\alpha
\end{equation*}
which can only be true if
\begin{equation*}
  \lambda_2(L) \le \frac{2c}{N}.
\end{equation*}
Thus, we have proven the following theorem.
\begin{theorem}\label{thm:1}
  \label{thm:uniform-bound}
  For the two-layer problem described above, we have the bound
  \begin{equation}
    \label{eq:uniform-bound}
      F(c) \le \frac{2c}{N}.
  \end{equation}
\end{theorem}

Now we turn our attention to the question of attainability of~\eqref{eq:uniform-bound}.  This question is answered by the following theorem.

\begin{theorem}\label{thm:2}
  \label{thm:achieve-bound}
  The inequality in~\eqref{eq:uniform-bound} can only be satisfied as equality if $w=\frac{c}{N}\ones$.
\end{theorem}

\begin{proof}
  Suppose the weights $w$ are chosen such that the Laplacian $L$ satisfies $\lambda_2(L)=2\frac{c}{N}$.  Then~\eqref{eq:reduced-inequality} simplifies to
  \begin{equation*}
    \begin{split}
      u_1^TL_1u_1 &+ u_2^TL_2u_2 + 4\alpha w^T(u_1-u_2)\\
      & + (u_1-u_2)^TW(u_1-u_2) \\
      &- \frac{2c}{N}\left(\|u_1\|^2 + \|u_2\|^2\right) \ge 0
    \ \ \forall \  \alpha, u_1^T\ones=u_2^T\ones=0.
    \end{split}
  \end{equation*}
  This can only be true if the linear coefficient in $\alpha$, $4w^T(u_1-u_2)$, vanishes for every choice of $u_1,u_2$ satisfying $u_1^T\ones = u_2^T\ones = 0$.  This implies that $w$ is parallel to $\ones$ and, since $w^T\ones = c$, the theorem follows.
\end{proof}

The previous theorem shows that when the bound~\eqref{eq:uniform-bound} is attained, it can only be attained by the uniform choice of weights $w=\frac{c}{N}\ones$.  The next theorem characterizes exactly the budgets for which the bound is attained.

\begin{theorem}\label{thm:3}
  For a given two-layer network, define the constant
  \begin{equation}
    \label{eq:cstar}
    \begin{split}
    c^* := &N \min_{\substack{u_1^T\ones=u_2^T\ones=0\\u_1 + u_2\ne 0}}
      \frac{u_1^TL_1u_1 + u_2^TL_2u_2}{\|u_1+u_2\|^2}\\
      \end{split}
  \end{equation}
  Then, for all budgets $c\ge 0$, $F(c)=\frac{2c}{N}$ if and only if $c\le c^*$.
\end{theorem}

\begin{proof}
  In light of Theorem~\ref{thm:achieve-bound}, we restrict our attention to the case of uniform weights $w=\frac{c}{N}\ones$ and use the notation $L=L(c)=L(w)$.  From Theorem~\ref{thm:achieve-bound}, we see that $F(c)=\frac{2c}{N}$ if and only if $\lambda_2(L(c))=\frac{2c}{N}$.  It is straightforward to check that $\frac{2c}{N}$ is an eigenvalue of $L(c)$ for any $c\ge 0$, with eigenvector $(\ones^T,-\ones^T)^T$.  Since $L(c)$ is positive semi-definite and $\lambda_1(L(c))=0$ for all $c$, it follows that $\lambda_2(L(c))\le\frac{2c}{N}$ for all $c$.  Thus, we have $F(c)=\frac{2c}{N}$ if and only if $\lambda_2(L(c))\ge \frac{2c}{N}$.  Now, as before, we write a vector $v$ orthogonal to $\ones$ as
  \begin{equation*}
    v =
    \begin{bmatrix}
      v_1\\v_2
    \end{bmatrix}
    =
    \begin{bmatrix}
      \alpha\ones + u_1\\
      -\alpha\ones + u_2
    \end{bmatrix},\quad u_1^T\ones = u_2^T\ones = 0.
  \end{equation*}
  Recalling~\eqref{eq:lambda2}, we observe that $\lambda_2(L(c))\ge\frac{2c}{N}$ if and only if the following inequality holds for every such choice of $v$ (or, equivalently, every choice of $\alpha$, $u_1$ and $u_2$).
  \begin{equation*}
    \begin{split}
      0 &\le v^TLv - \frac{2c}{N}\|v\|^2 \\
      & = v_1^TL_1v_1 + v_2^TL_2v_2 + \frac{c}{N}\|v_1-v_2\|^2
      - \frac{2c}{N}\left(\|v_1\|^2 + \|v_2\|^2\right) \\
      &= u_1^TL_1u_1 + u_2^TL_2u_2 - \frac{c}{N}\|u_1+u_2\|^2.
    \end{split}
  \end{equation*}
  This inequality holds for all $u_1^T\ones=u_2^T\ones=0$ if and only if $c\le c^*$ as defined in~\eqref{eq:cstar}, completing the proof.

\end{proof}
The threshold obtained by Eq. \eqref{eq:cstar} is exactly equivalent to the threshold found in \cite{sahneh2014exact} (see Supplemental Material):
\begin{equation}
\label{eq:exact}
c^* = N\lambda_{2}\left(\left(L_1^\dagger + L_2^\dagger\right)^\dagger\right),
\end{equation}
where $L^\dagger$ represents the Moore-Penrose pseudoinverse of $L$.
At the threshold a rough lower-bound for $\lambda_2(L)$ is
\begin{equation}
    \label{ineq:threshold}
 \lambda_2(L)=\frac{2}{N}c^*\geq \min\{\lambda_2(L_1),\lambda_2(L_2)\}.
\end{equation}
One way to see this is to observe that:
\begin{equation*}
\frac{u_1^TL_1u_1 + u_2^TL_2u_2}{\|u_1+u_2\|^2} \geq    \frac{\|u_1\|^2+\|u_2\|^2}{\|u_1+u_2\|^2}\min\{\lambda_2(L_1),\lambda_2(L_2)\}.
\end{equation*}
Inequality~\eqref{ineq:threshold} then follows from the parallelogram law \cite{cantrell2000modern}.
An upper bound for $\lambda_2(L)$ is given in~\cite{gomez2013diffusion}
\begin{equation}
  \label{eq:gomez-bound}
  \lambda_2(L) \le \frac{1}{2}\lambda_2(L_1+L_2).
\end{equation}
In the special case of identical layers ($L_1=L_2$) with corresponding nodes connected, the bound in \eqref{eq:gomez-bound} is attained with uniform weights at the threshold budget $c^*$~\cite{Radicchi2013}. This can be seen by combining (\ref{ineq:threshold}) and (\ref{eq:gomez-bound}).  Therefore, in this case, uniform weights are optimal for budgets $c\le c^*$, and increasing the budget beyond $c^*$  cannot increase the algebraic connectivity, regardless of the weight allocation.



For general structures, it is possible to substantially improve the algebraic connectivity by increasing the budget beyond $c^*$ using an optimal weight distribution.  Figs.~\ref{fig:CompareUniformOptimal}a and~\ref{fig:CompareUniformOptimal}b compare the optimal value of $\lambda_2(L)$ to the one obtained by the uniform distribution as the budget $c$ varies for two different network structures. In both cases, the optimal distribution gives a higher algebraic connectivity after the threshold.
 \begin{figure*}
\subfloat[]{%
  \includegraphics[clip,width=.8\columnwidth]{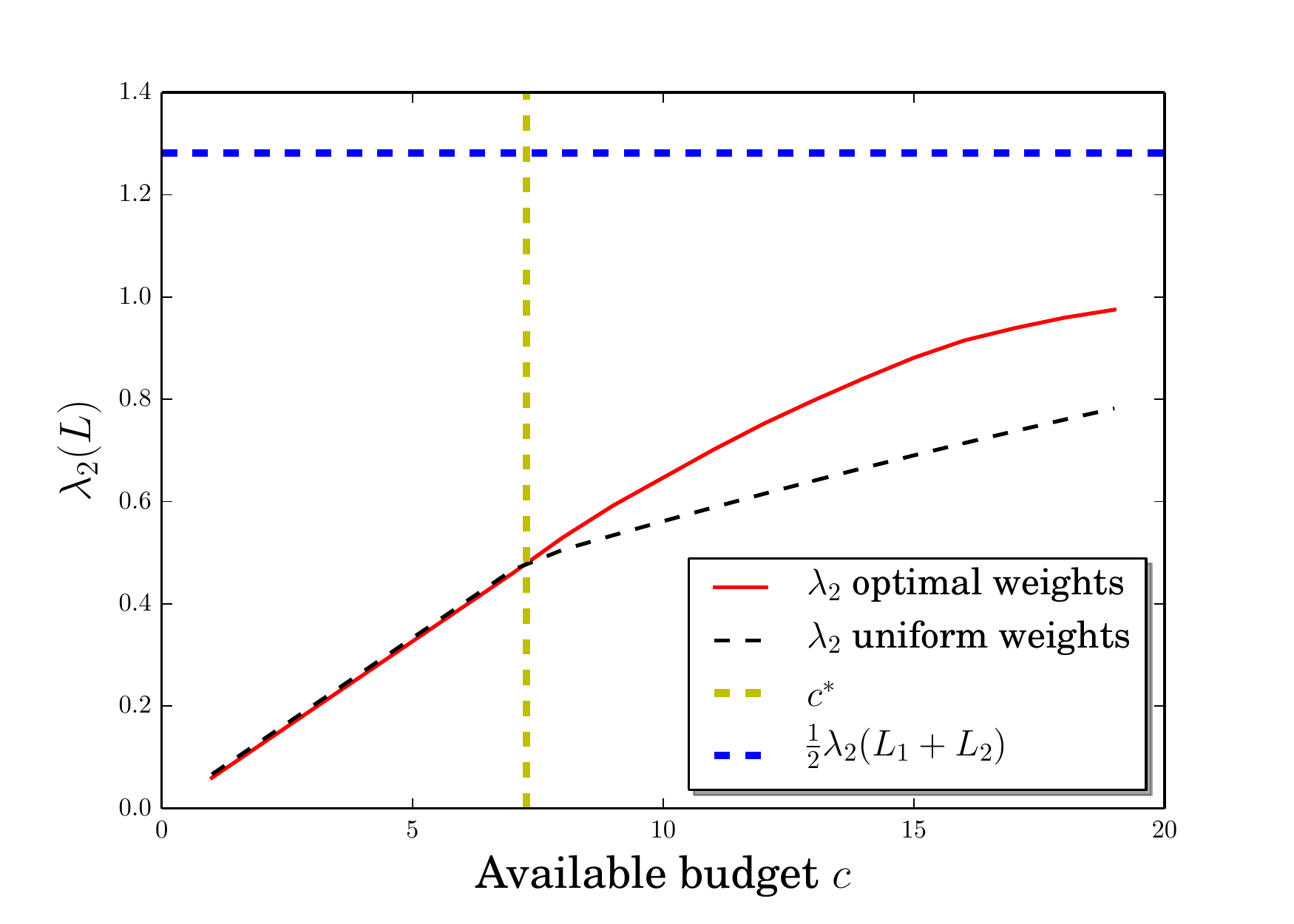}%
}
\subfloat[]{%
  \includegraphics[clip,width=.8\columnwidth]{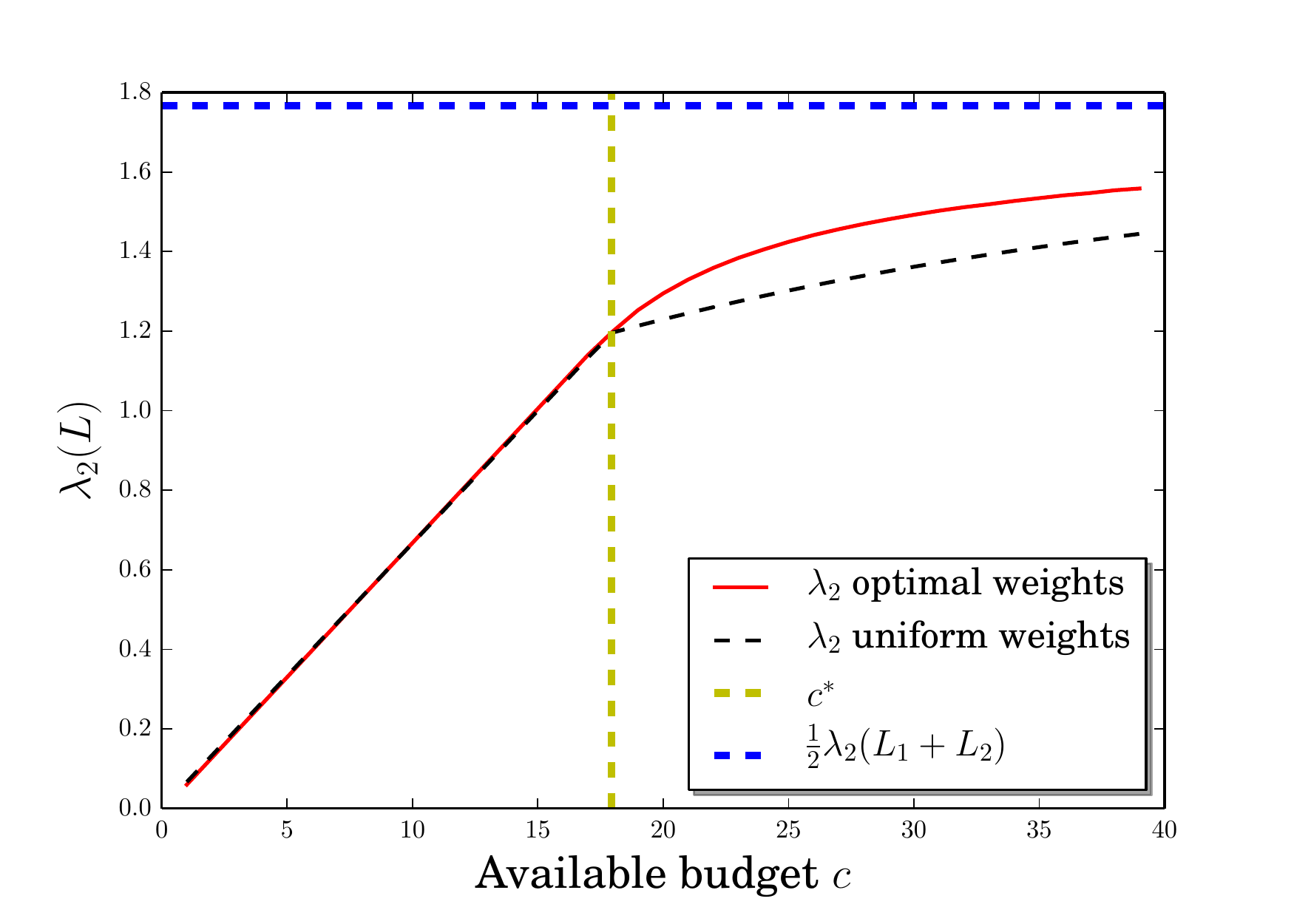}%
}

\subfloat[]{%
  \includegraphics[clip,width=.8\columnwidth]{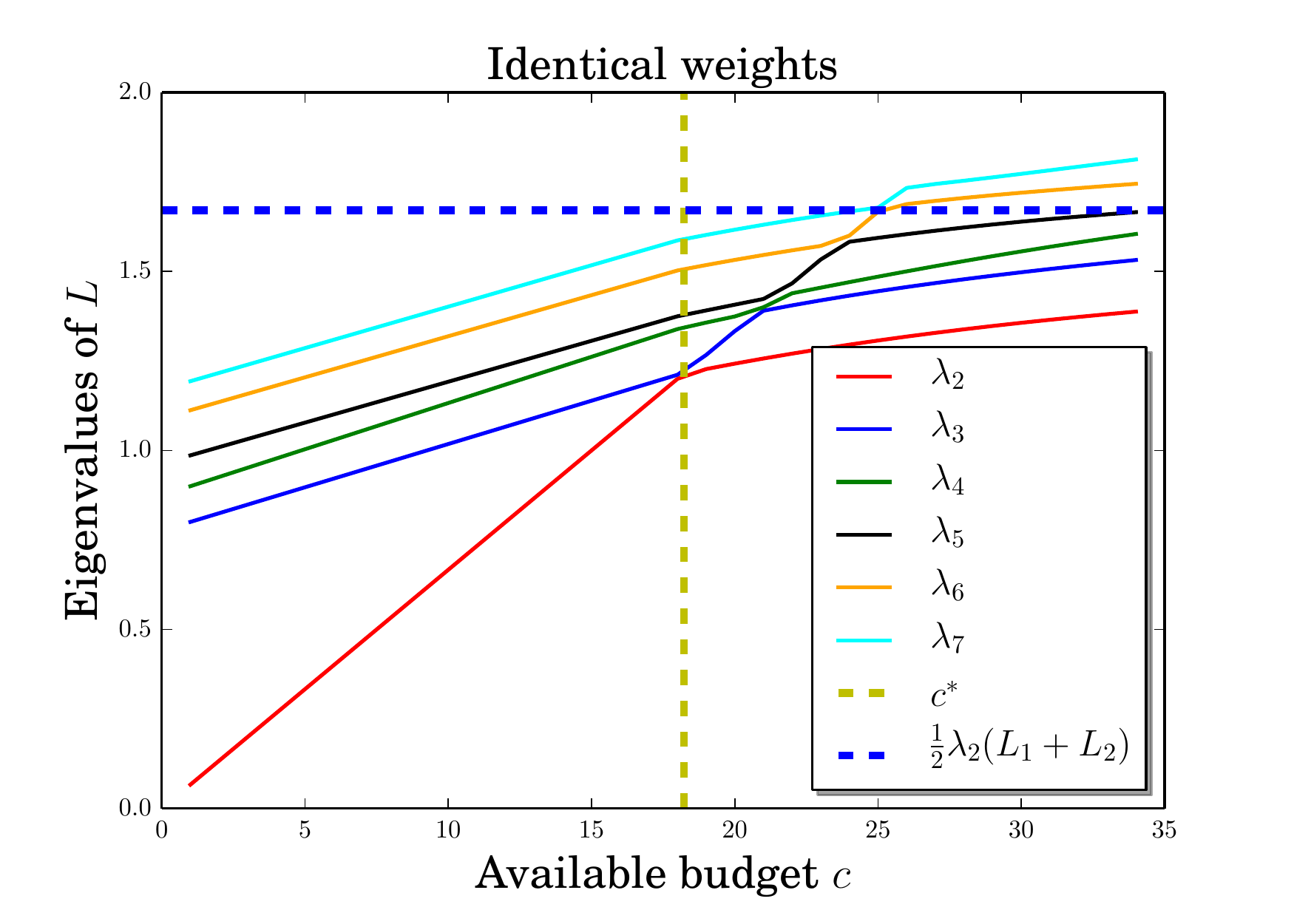}%
}
\subfloat[]{%
  \includegraphics[clip,width=.8\columnwidth]{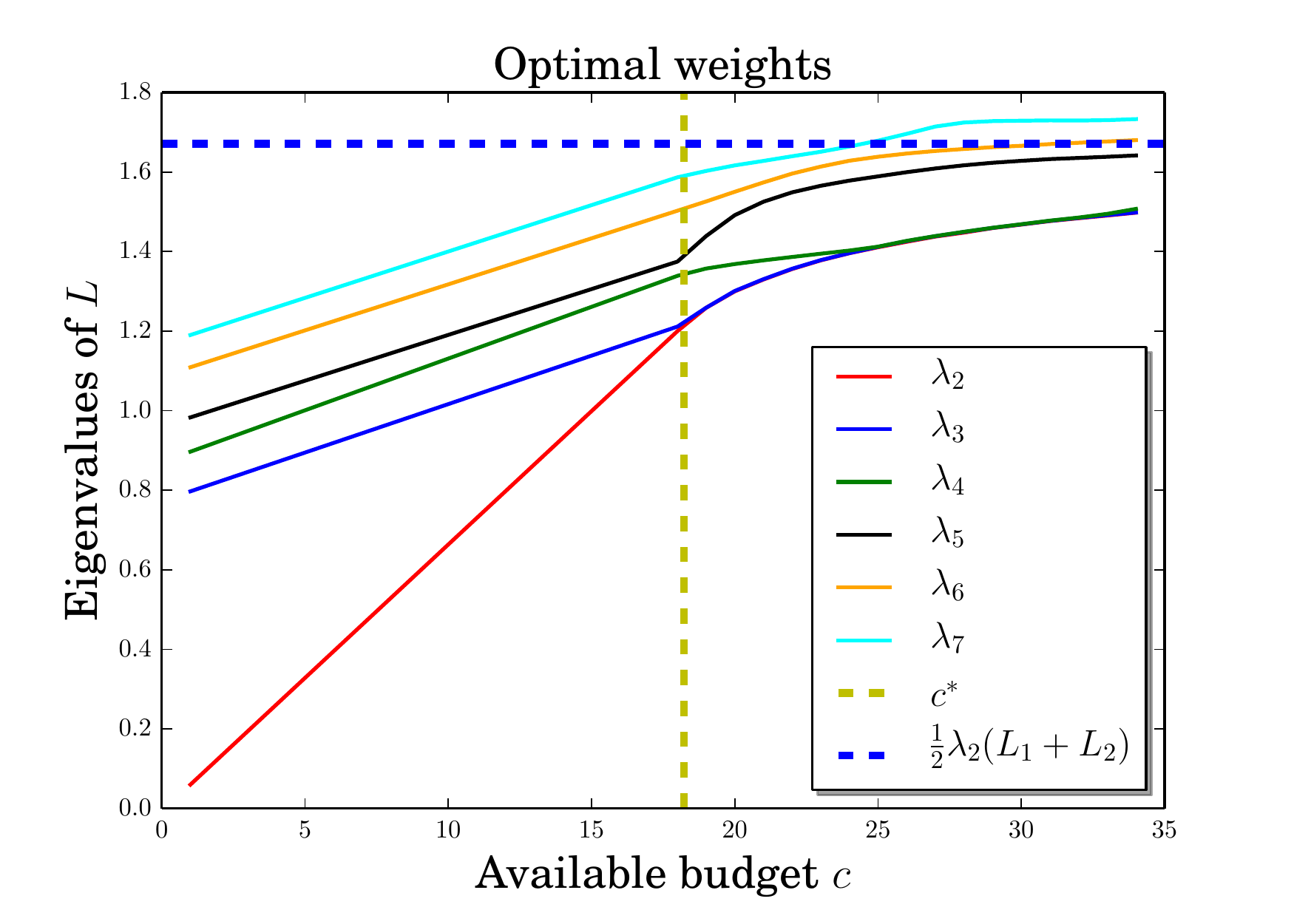}%
}
\caption{(a) and (b) Plots of $\lambda_2(L)$ with different amount of available budget. The solid red line is for the optimal weights and the dashed black line is for uniform weights. The threshold budget and upper-bound is shown with yellow and blue dashed lines respectively. 
The upper-bound is from Eq.~\eqref{eq:gomez-bound} and the threshold is from Eq.~\eqref{eq:exact}. 
(a) A structure of two Erd\"{o}s-Renyi networks each with $30$ nodes and (b) a structure of two scale-free networks each with $30$ nodes. (c) First seven eigenvalues of Laplacian matrix of $\GG$ considering a uniform distribution of weights for the multiplex in (b). 
(d) First seven eigenvalues of Laplacian matrix of $\GG$ considering an optimal distribution of weights for the multiplex in (b).}
\label{fig:CompareUniformOptimal}
\end{figure*}
 

In Fig. \ref{fig:CompareUniformOptimal}c, we plot the first seven eigenvalues of $L$ (omitting the zero eigenvalue) for a multiplex with identical weights on the inter-layer links. Because $\frac{2c}{N}$ is always an eigenvalue and $\lambda_3(L)>\frac{2c}{N}$ for $c\rightarrow 0$, increasing $c$, $\lambda_2(L)$ and $\lambda_3(L)$ cross.
For the same multiplex with optimal distribution of inter-layer weights, we plot the eigenvalues in Fig. \ref{fig:CompareUniformOptimal}d.
When increasing the budget beyond the threshold, the second and third eigenvalues
 coalesce and are strictly less than $\frac{2c}{N}$.
Since \eqref{eq:F} is a convex optimization problem, we know the optimal $w_i$'s vary continously with $c$, and smooothly 
away from the finite set of budgets where eigenvalue multiplicities change. 


When $c\leq c^*$, the Fiedler vector is $v = \frac{1}{\sqrt{2N}}[\ones, -\ones]$ and the Fiedler cut distinguishes the layers. For $c>c^*$, due to the multiplicity of $\lambda_2(L)$, there is a corresponding Fiedler eigenspace.  Therefore, the two layers are not as easily recognizable as before.

In Fig. \ref{fig:CompareUniformOptimal}, we also observe that  for $c>c^*$,  $\lambda_2$ increases more slowly. Moreover, as Fig.~\ref{fig:WD} shows, we can have very non-uniform weights in this case.
\begin{figure*}
\subfloat[]{%
  \includegraphics[clip,width=.53\columnwidth]{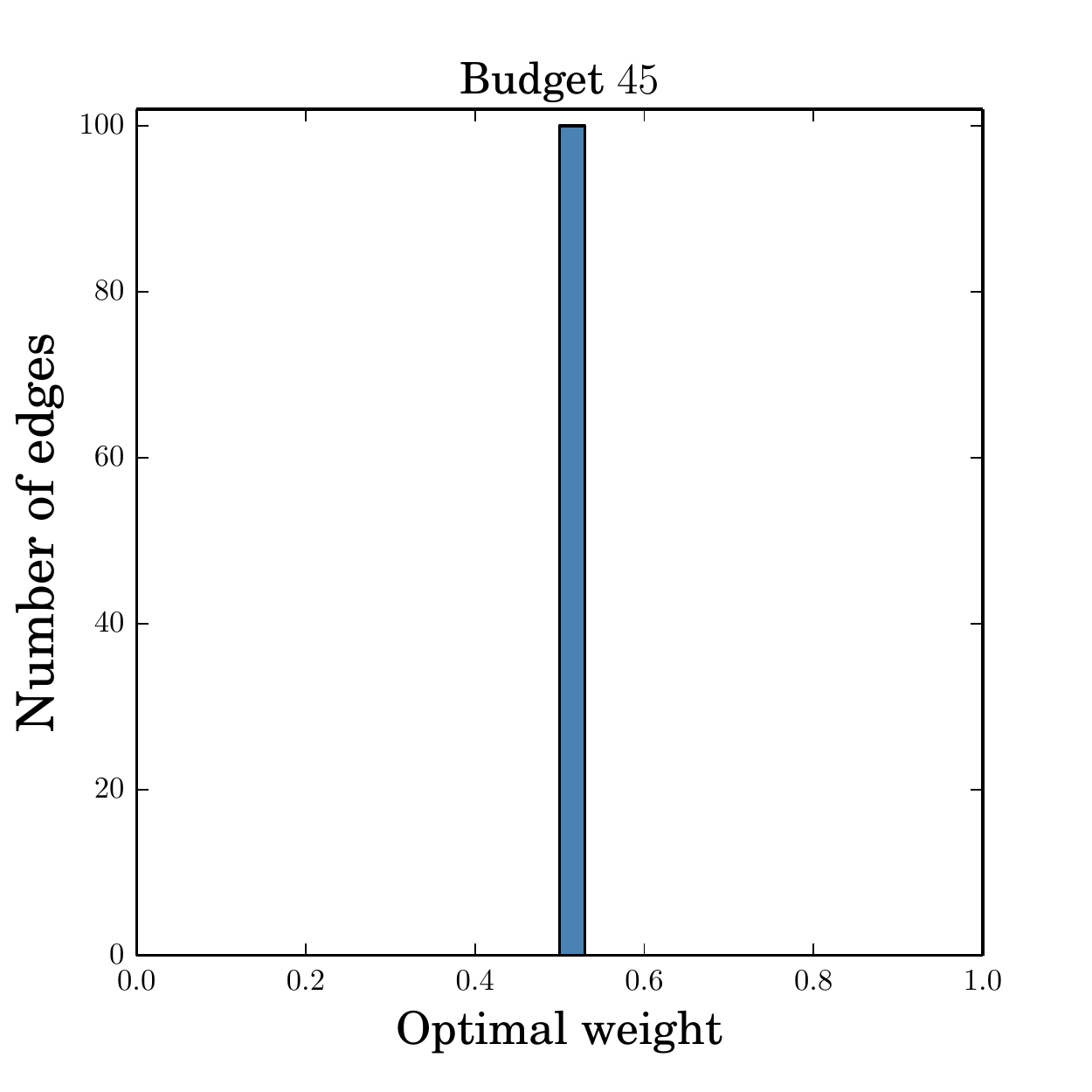}%
}
\subfloat[]{%
  \includegraphics[clip,width=.53\columnwidth]{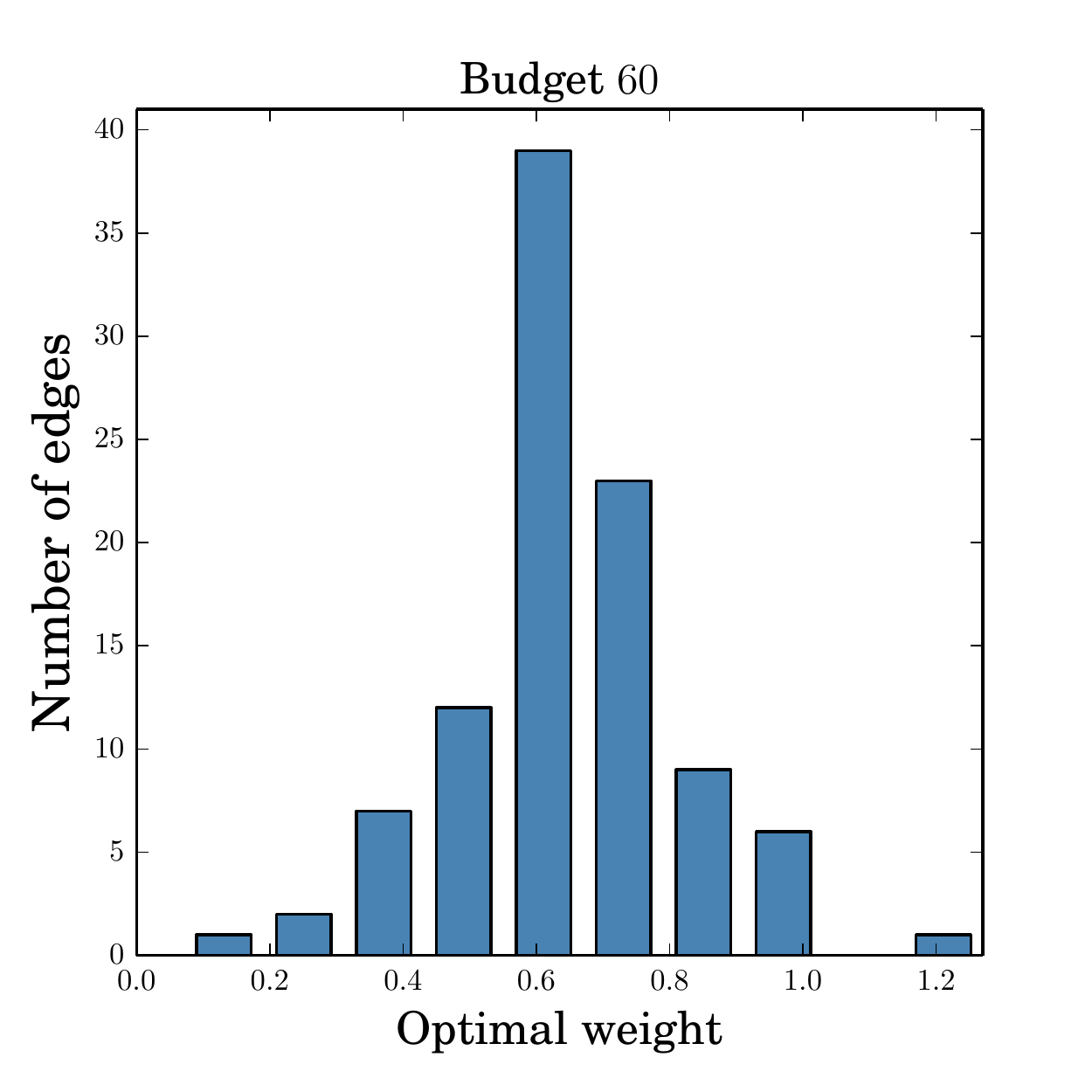}%
}
\subfloat[]{%
  \includegraphics[clip,width=.53\columnwidth]{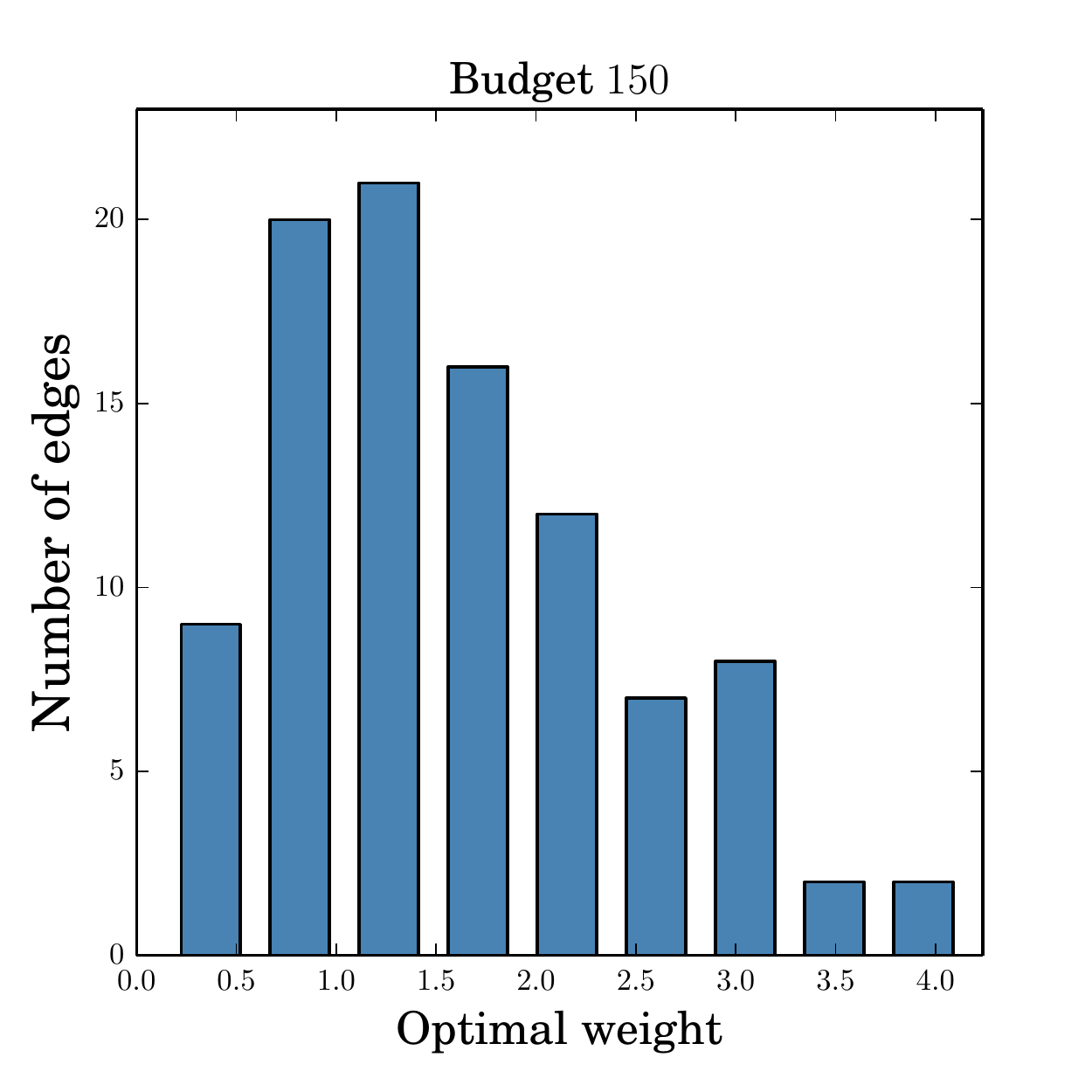}%
}
\subfloat[]{%
  \includegraphics[clip,width=.53\columnwidth]{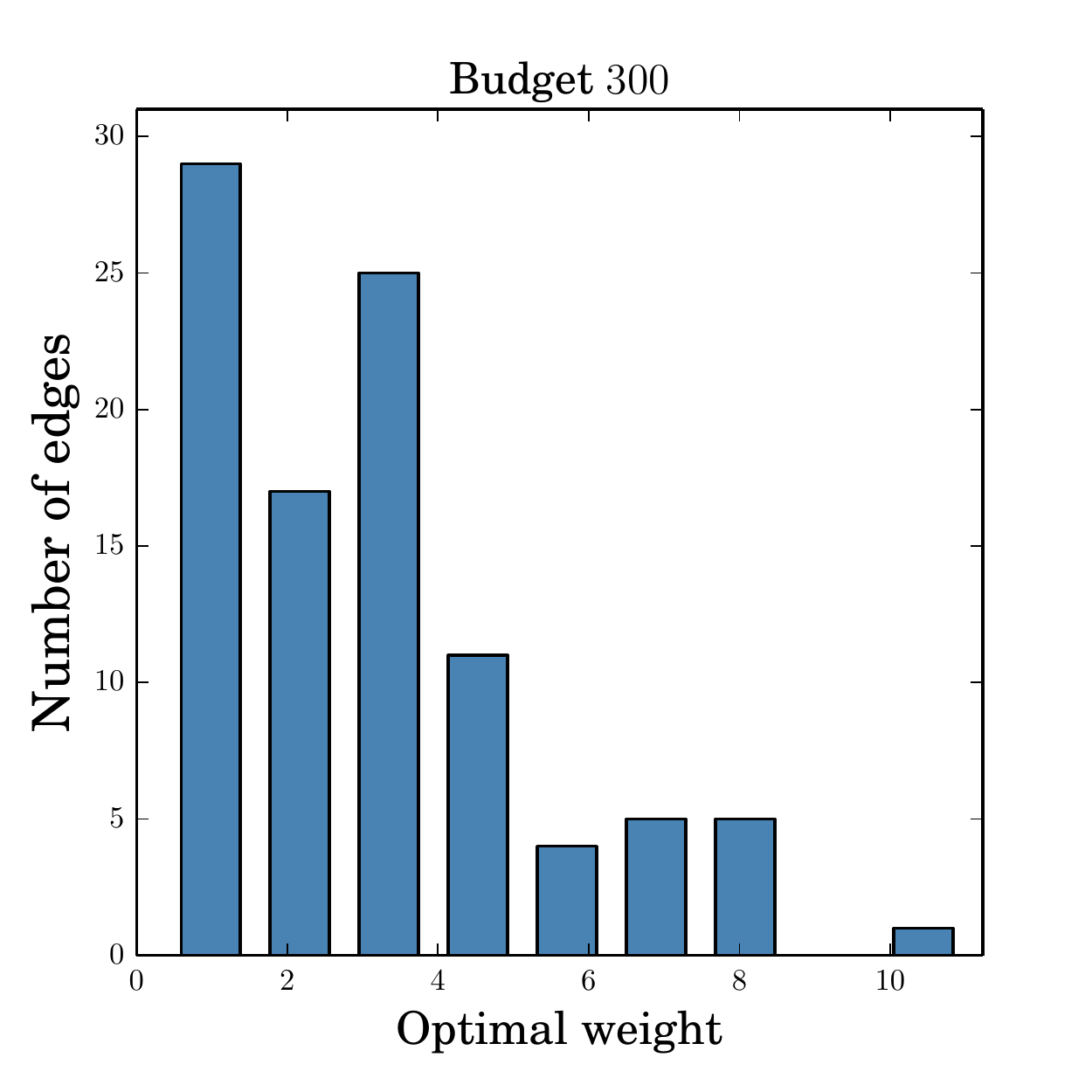}%
}
\caption{Optimal weight distribution for different amount of budgets. The stucture of a multiplex with two scale free network layers, with $N = 100$ nodes and $\Vert E_1\Vert =  196$ and $\Vert E_2\Vert = 291$. In (a) budget is lower than threshold and uniform distribution is optimal. In this example, the threshold budget $c^*$ is $51.4$.}\label{fig:WD}
\end{figure*}
These optimal weights  represent  the importance  of each link in improving the algebraic connectivity of the whole network. 
%
%
 
%
In summary, we have shown that before a threshold budget that can be analytically computed, the largest possible algebraic connectivity is a linear function of the budget and can  only be attained by the uniform weight distribution. Since the threshold budget is always strictly positive, for low enough budgets it is not necessary to solve \eqref{eq:F}. On the other hand, for larger budgets,  \eqref{eq:F} can be solved with efficient semi-definite programming solvers to find the optimal weights. In particular, heuristic methods based solely on the information of each layer are too  
blunt to notice this threshold phenomenon. 
\bibliographystyle{apsrev4-1}
\bibliography{Lap_Eig_Letter}
\onecolumngrid
\appendix
\section*{Supplementary Material}

We have defined the threshold budget $c^*$ as
\begin{equation}
    \label{eq:threshold}
    c^* = N \min_{\substack{u_1^T\ones=u_2^T\ones=0\\u_1+u_2\ne 0}}
    \frac{u_1^TL_1u_1 + u_2^TL_2u_2}{\|u_1+u_2\|^2},
\end{equation}
where $L_1$ and $L_2$ are the Laplacian matrices for the two individual layers.  

\begin{theorem}
The threshold budget $c^*$ defined in \eqref{eq:threshold} satisfies
    \begin{equation}
        \frac{c^*}{N} = \lambda_{2}\left(\left(L_1^\dagger + L_2^\dagger\right)^\dagger\right) 
    \end{equation}
\end{theorem}

\begin{proof}
    We begin by rewriting the minimization in~\eqref{eq:threshold}:
    \begin{equation}
    \label{eq:split-min}
    \begin{split}
        \min_{\substack{u_1^T\ones=u_2^T\ones=0\\u_1+u_2\ne 0}}
        \frac{u_1^TL_1u_1 + u_2^TL_2u_2}{\|u_1+u_2\|^2}
        &=
        \min_{\substack{u^T\ones=0\\u\ne 0}}
        \min_{\substack{u_1^T\ones=u_2^T\ones=0\\u_1+u_2=u}}
        \frac{u_1^TL_1u_1 + u_2^TL_2u_2}{\|u\|^2}\\
        &=
        \min_{\substack{u^T\ones=0\\u\ne 0}}
        \frac{1}{\|u\|^2}
        \min_{\substack{u_1^T\ones=u_2^T\ones=0\\u_1+u_2=u}}
        \left(u_1^TL_1u_1 + u_2^TL_2u_2\right).
    \end{split}
    \end{equation}
    
    To solve the inner minimization problems, we introduce Lagrange multipliers to find that the minimizing $u_1$ and $u_2$ satisfy
    \begin{equation*}
        L_1u_1 = \nu\ones + \mu,\qquad
        L_2u_2 = \eta\ones + \mu.
    \end{equation*}
    Taking an inner product of each of these with the $\ones$ vector shows that
    \begin{equation*}
    \nu = \eta = - \frac{\mu^T\ones}{N}, 
    \end{equation*}
    so that
    \begin{equation*}
    u_1 = L_1^\dagger\left(\mu - \frac{\mu^T\ones}{N}\right),\qquad
    u_2 = L_2^\dagger\left(\mu - \frac{\mu^T\ones}{N}\right).
    \end{equation*}
    Thus, without loss of generality, $\mu$ can be taken to be orthogonal to $\ones$.  With this form, $u_1$ and $u_2$ are already orthongal to $\ones$ as well.  In order to satisfy the constraint $u_1+u_2=u$, we must have
    \begin{equation*}
        \left(L_1^\dagger + L_2^\dagger\right)\mu = u,\quad\text{i.e.},\quad
        \mu = \left(L_1^\dagger + L_2^\dagger\right)^\dagger u.
    \end{equation*}
    
    From this, we see that the minimizing $u_1$ and $u_2$ of the inner minimization problem in~\eqref{eq:split-min} satisfy
    \begin{equation*}
        u_1 = L_1^\dagger\left(L_1^\dagger + L_2^\dagger\right)^\dagger u,\qquad
        u_2 = L_2^\dagger\left(L_1^\dagger + L_2^\dagger\right)^\dagger u,
    \end{equation*}
    giving a minimum value of
    \begin{equation}
    \label{eq:rayleigh}
    \begin{split}
        u_1^TL_1u_1 + u_2^TL_2u_2 
        &= u^T\left(L_1^\dagger + L_2^\dagger\right)^\dagger L_1^\dagger L_1 L_1^\dagger\left(L_1^\dagger + L_2^\dagger\right)^\dagger u +\\
        &\qquad u^T\left(L_1^\dagger + L_2^\dagger\right)^\dagger L_2^\dagger L_2L_2^\dagger\left(L_1^\dagger + L_2^\dagger\right)^\dagger u\\
        &= u^T\left(L_1^\dagger + L_2^\dagger\right)^\dagger L_1^\dagger\left(L_1^\dagger + L_2^\dagger\right)^\dagger u +\\
        &\qquad u^T\left(L_1^\dagger + L_2^\dagger\right)^\dagger L_2^\dagger\left(L_1^\dagger + L_2^\dagger\right)^\dagger u\\
        &= u^T\left(L_1^\dagger + L_2^\dagger\right)^\dagger\left(L_1^\dagger + L_2^\dagger\right)\left(L_1^\dagger + L_2^\dagger\right)^\dagger u\\
        &= u^T\left(L_1^\dagger + L_2^\dagger\right)^\dagger u.
    \end{split}
    \end{equation}
    Here, we have used the identity $A^\dagger AA^\dagger=A^\dagger$.
    
    Substituting back into~\eqref{eq:split-min}, we have
    \begin{equation*}
        \frac{c^*}{N} = \min_{\substack{u^T\ones=0\\u\ne 0}}
            \frac{u^T\left(L_1^\dagger + L_2^\dagger\right)^\dagger u}{\|u\|^2}.
    \end{equation*}
    Since $L_1$ and $L_2$ are positive semidefinite, so are $L_1^\dagger$ and $L_2^\dagger$ and, consequently, so are $L_1^\dagger + L_2^\dagger$ and $\left(L_1^\dagger + L_2^\dagger\right)^\dagger$.  Since the component networks are assumed connected, the nullspace of $\left(L_1^\dagger + L_2^\dagger\right)^\dagger$ is spanned by the vector $\ones$.  The Rayleigh quotient in~\eqref{eq:rayleigh} is therefore minimized over the orthogonal complement of the eigenspace associated with the first eigenvalue of $\left(L_1^\dagger + L_2^\dagger\right)^\dagger$ and the theorem follows.
\end{proof}
\end{document}